%% file: simple_hc.tex
\documentclass[11pt]{article}

\usepackage{amsmath}
\usepackage{amsfonts}
\usepackage{amssymb}
\usepackage{amsthm}
\usepackage{enumerate}
\usepackage{hyperref}
\usepackage{geometry}
\usepackage{xcolor}
\usepackage{xspace}
\usepackage{todonotes}

\newtheorem{theorem}{Theorem}[section]
\newtheorem{lemma}[theorem]{Lemma}
\newtheorem{claim}[theorem]{Claim}
\newtheorem{definition}[theorem]{Definition}
\newtheorem{corollary}[theorem]{Corollary}

\newtheorem{remark}[theorem]{Remark}

\newcommand{\suchthat}{\;\ifnum\currentgrouptype=16 \middle\fi|\;}

\newcommand{\N}{\mathbb{N}}

\newcommand*{\defeq}{\stackrel{\text{def}}{=}}

\DeclareMathOperator{\poly}{poly}
\DeclareMathOperator{\agreement}{ag}
\newcommand{\E}{{\rm I\kern-.3em E}}

\newcommand{\p}{\ensuremath{\mathsf{P}}}
\newcommand{\np}{\ensuremath{\mathsf{NP}}}

\newcommand{\hc}{\textsc{HypCol}}
\newcommand{\nph}{{\np}-{hard}\xspace}
\newcommand{\etal}{\textit{et al}.~}

\newcommand{\calH}{\mathcal{H}}
\newcommand{\HV}{\mathcal{V}}
\newcommand{\HE}{\mathcal{E}}

\newcommand{\calI}{{\mathcal I}}

\hypersetup{
  citecolor={blue}
}

\begin{document}
\title{Simplified inpproximability of hypergraph coloring via $t$-agreeing families}
\author{
Per Austrin\thanks{ KTH Royal Institute of Technology. {\tt austrin@kth.se}.  Research supported by the Approximability and Proof Complexity project funded by the Knut and Alice Wallenberg Foundation.} \and
Amey Bhangale \thanks{Weizmann Institute of Science, Rehovot, Israel.  {\tt amey.bhangale@weizmann.ac.il}. Research supported by Irit Dinur's ERC-CoG grant 772839.}\and
Aditya Potukuchi \thanks{Department of Computer Science,  Rutgers University. {\tt aditya.potukuchi@cs.rutgers.edu}. Research supported by Mario Szegedy\rq{}s NSF grant 
CCF-1514164}
}

\maketitle

\begin{abstract}
We reprove the results on the hardness of approximating hypergraph coloring using a different technique based on bounds on the size of extremal $t$-agreeing families of $[q]^n$.  Specifically, using theorems of Frankl-Tokushige~\cite{FT99}, Ahlswede-Khachatrian~\cite{AK98} and Frankl~\cite{F76} on the size of such families, we give simple and unified proofs of quasi \np-hardness of the following problems:
\begin{itemize}
\item coloring a $3$ colorable $4$-uniform hypergraph with $(\log n)^\delta$ many colors
\item coloring a $3$ colorable $3$-uniform hypergraph with $\tilde{O}(\sqrt{\log \log n})$ many colors
\item coloring a $2$ colorable $6$-uniform hypergraph with $(\log n)^\delta$ many colors
\item coloring a $2$ colorable $4$-uniform hypergraph with $\tilde{O}(\sqrt{\log \log n})$ many colors
\end{itemize}
where $n$ is the number of vertices of the hypergraph and $\delta>0$ is a universal constant. 
\end{abstract}

 \section{Introduction}
 
We study the fundamental problem of coloring hypergraphs with minimum number of colors. A $k$-uniform hypergraph $H(V, E)$ consists of a collection of vertices $V$ and a set of hyperedges $E \subseteq \binom{V}{k}$. 

A coloring $\chi : V\rightarrow [c]$ is called a {\em proper coloring} of $H$ if no hyperedge $e\in E$ is monochromatic with respect to the coloring $\chi$. The chromatic number of a hypergraph is the minimum number of colors needed to properly color it.

In this paper we study the problem of {\em approximating} the chromatic number of a given hypergraph. More specifically, we study the following problem: Given a $c$-colorable $k$-uniform hypergraph, find a proper $D$-coloring of it in polynomial time for a given $D\geq c$.  There has been rich history of studying the computational complexity of finding or approximating the chromatic number of graphs as well as hypergraphs. It is known that unless $\p=\np$, approximating the chromatic number of an $n$ vertex graph to within a factor of $n^{1-\varepsilon}$ is hard for all $\varepsilon>0$~\cite{FK98, Z06}. The best approximation algorithms currently known for chromatic number give an approximation factor guarantee of $O\!\left(\frac{n  (\log \log n)^2}{(\log n)^3}\right)$~\cite{H93}.

Given these results, a lot of attention has been devoted to understanding the complexity of finding a proper coloring given the guarantee that the hypergraph has very small, even constant, chromatic number. In fact, it is \np-hard to decide if a given hypergraph is $2$-colorable or not, unlike the graph case where it is easy to decide if the graph is bipartite or not in polynomial time. The best polynomial time algorithms currently known require $n^{\Omega(1)}$ colors to color a $2$-colorable hypergraph~\cite{KNS01, CF96, AKMR96, KT17}.
Since finding a proper $D$-coloring is at least as hard as finding the independent set of size $n/D$, a lot of attention went into studying the following (computationally easier) problem:

\begin{definition}[$\hc_n(k, c, D)$]
Given a $k$-uniform hypergraph $H(V, E)$ on $n$ vertices which is $c$-colorable, find an independent set of size $\frac{n}{D}$.
\end{definition}

\input{table}

The study of the complexity of approximate hypergraph coloring was initiated by Guruswami \etal \cite{GHS02}. Holmerin~\cite{H02} and  Guruswami \etal \cite{GHS02} showed that  \linebreak $\hc_n(4,2, \frac{\log \log n}{\log \log \log n})$ is quasi \nph\footnote{Recall that if a problem is quasi \nph then it cannot be solved in quasipolynomial time $2^{\poly \log n}$ unless all problems in $\np$ can be solved in quasipolynomial time.}. Khot~\cite{Khot02_33, Khot02_q4} showed quasi \np-hardness of $\hc_n(4,q, (\log n)^{cq})$ for some $c>0$ and all $q\geq 7$ and $\hc_n(3,3, O(\log \log n)^{1/9})$. Dinur-Guruswami~\cite{DG15} showed $\hc_n(6, 2, (\log n)^c)$ is quasi \nph. Saket~\cite{S14} improved the state for 2-colorable $4$-uniform hypergraph by showing $\hc_n(4, 2, (\log n)^c)$ is quasi \nph for some $c>0$.

Guruswami \etal \cite{GHHSV17} broke the logarithmic barrier for the first time and showed quasi \np-hardness of $\hc_n(8, 2, 2^{2^{\sqrt{\log \log n}}})$, $\hc_n(4, 4, $ $2^{2^{\sqrt{\log \log n}}})$ and  
 $\hc_n(3, 3, 2^{\frac{\log \log n}{\log \log \log n}})$ using the {\em short code}. Finally, Khot and Saket \cite{KS14} improved the factor to almost polynomial by showing quasi \np-hardness of $\hc_n(12, 2, 2^{(\log n)^{\Omega(1)}})$. Building on the work of  Khot and Saket~\cite{KS14}, Varma~\cite{Varma15} showed quasi NP-hardness of $\hc_n(8, 2, 2^{(\log n)^{\Omega(1)}})$ as well as $\hc_n(4, 4, $ $2^{(\log n)^{\Omega(1)}})$.  In the $2$-colorable case, Huang~\cite{Huang15} independently showed quasi \np-hardness of $\hc_n(8, 2, 2^{(\log n)^{1/20 - o(1)}})$.

In terms of approximating the chromatic number for $2$ colorable hypergraphs,  Dinur\etal~\cite{DRS02} showed that it is quasi \nph to color $2$-colorable $3$-uniform hypergraph by $(\log \log n)^{1/9}$. This result is weaker than showing hardness for $\hc_n(3, 2, (\log \log n)^{1/9})$. In fact, it is still open to determine the complexity of $\hc_n(3, 2, c)$ for any  $c=\omega(1)$. Very recently, the second author~\cite{B18} showed NP-hardness of coloring $2$-colorable $4$-uniform hypergraph with $(\log n)^c$ colors for some $c>0$.

\subsection{Our results}

We give unified proofs of many of the known results on the hardness of approximate hypergraph coloring. Our analysis makes a novel use of the maximum size of $t$-agreeing families. 

We now state the theorems that we (re)prove. See Section~\ref{section:comparison} for the comparison with the previous works on hypergraph coloring.

Our first result gives an alternate proof of the result of Khot~\cite{Khot02_q4} for $q$-colorable $4$-uniform hypergraph for smaller values of $q$.
\begin{theorem}
\label{theorem:three_four}
There exists a constant $\delta>0$ such that $\hc_n(4, 3, \log^\delta n)$ is quasi \nph.
\end{theorem}

For the $3$-uniform hypergraph, Khot~\cite{Khot02_33} and Dinur \etal~\cite{DRS02} start with a multi-layered Label Cover instance, which was one of the highlights of their proofs. In our proof, we also start with this multi-layered Label Cover instance, but simplify\footnote{modulo the theorem on $t$-agreeing family, which we use as a black box.} the inner verification step. Our proof is almost along the lines of proof of Dinur \etal ~\cite{DRS02}, but we get a stronger independent set guarantee recovering (slightly improving) the result of Khot~\cite{Khot02_33}.
\begin{theorem}
\label{theorem:three_three}
$\hc_n(3, 3, \tilde{O}(\sqrt{\log\log n}))$ is quasi \nph.
\end{theorem}
We note that the above theorem gives a weaker bound than the result of ~\cite{GHHSV17}, $\hc_n(3, 3, 2^{\frac{\log \log n}{\log \log \log n}})$,  which is the best known result for $3$-colorable $3$-uniform hypergraph.\\

Both the previous results have in completeness case a hypergraph which is $3$-colorable. Our next theorem shows hardness of finding independent sets for $2$-colorable hypergraphs, but with slightly larger uniformity. This gives another proof of ~\cite{DG15}.

\begin{theorem}
\label{theorem:two_six}
There exists a constant  $\delta>0$ such that  $\hc_n(6, 2, \log^\delta n)$ is quasi \nph.
\end{theorem}

We also extend our techniques to prove the following hardness for $2$-colorable $4$-uniform hypergraph which, up to a quadratic factor, recovers the result of  \cite{GHS02} and~\cite{H02}.
\begin{theorem}
\label{theorem:two_four}
$\hc_n(4, 2, \tilde{O}(\sqrt{\log\log n}))$ is quasi \nph.
\end{theorem}
In this case too, Saket~\cite{S14} gets a better guarantee.

In addition to giving alternate proofs of the known results, our proofs give rise to interesting questions about {\em $t$-agreeing families}. If solved in a positive way could lead to improved inapproximability of hypergraph coloring for lower uniformity hypergraphs. See Section~\ref{section:conclusion} for more details.

%

\subsection{Proof overview}
For the proof overview, we are going to think of a Label Cover instance as a regular graph $G(V,E, [L])$ on $V$ with the alphabet $[L]$. The edges of the graph are labeled with $d$-to-$d$ constraints $\phi_e$ for every $e\in E$, for some $1\leq d\leq L$.  The instance is always {\em regular} which means that there exists $\delta>0$ such that any induced graph on $\alpha |V|$ vertices has at least $\Omega(\alpha)$ fraction of the constraints, for $\alpha \geq 1/L^{\delta}$.

We now give a proof overview of quasi-\np hardness of $\hc_n(4, 3, \log^\delta n)$. The starting point is the gap Label Cover problem where distinguishing between the cases when the Label Cover instance is satisfiable vs. no assignment can satisfy more than $L^{-\varepsilon}$ fraction of the edges is \nph, for some constant $\varepsilon>0$. Given an instance $G(V, E, L)$  of gap Label Cover, we reduce it to a $4$-uniform hypergraph as follows. We replace every vertex $v$ with a cloud $C[v]$ of size $3^L$, where the vertex in a cloud is referred by a pair $(v, a)$ for $a\in [3]^L$. Thus, the  number of vertices in the hypergraph is $|V|\cdot 3^L$. Now for every edge $e(u,v)$ in $G$ we put a hyperedge between $(u, a_u), (u,b_u)$ and $(v,a_v), (v,b_v)$ iff
 for every label $i$ and $j$ such that $(i,j)$ satisfies the constraint $\phi_e$, $\{a_u(i), b_u(i), a_v(j), b_v(j)\}$ are not all equal. 

The completeness case is easy: Given a labeling $\ell: V \rightarrow [L]$ to $G$ satisfying all the edges, we color $(v, a)$ with the color $a_{\ell(v)}$. It is easy to see that any hyperedge $\{(u, a_u), (u,b_u), (v,a_v), (v,b_v)\}$ will get at least two distinct colors {\em by construction}. 

Let  us consider the more interesting soundness case. Suppose the Label Cover instance is only $1/L^{\varepsilon}$ satisfiable. Suppose the hypergraph has an independent set $\calI$ of fractional size $\alpha\in [0,1]$. Then by simple averaging argument, there exist at least $\alpha/2$ fraction of the vertices $v\in V$ such that $|\calI \cap C[v]|\geq \frac{\alpha}{2}\cdot 3^L$ Let $X\subseteq V$ be the set of such vertices. For every $v\in X$, using a theorem on $t$-agreeing families (Theorem~\ref{thm:FT}), there must exist $a_v, b_v \in \calI \cap C[v]$ such that they agree on at most $t:=O(\log(1/\alpha))$ coordinates. Denote $L_v$ be the set of such coordinates. Now, for an edge $(u,v)$ such that both $u, v\in X$, it must be the case that there exist $i\in L_u$ and $j\in L_v$ such that $(i,j)$ satisfies $\pi_{(u,v)}$. In fact, if this was not the case then the $t$-agreeing pairs from the clouds $C[v]$ and $C[u]$ (which was used to define the sets $L_u$ and $L_v$) form a valid hyperedge! Thus, for every vertex $v\in X$, we have a small list of at most $t$ labels. If we assign a random label from $L_v$ to $v$ for all $v\in X$ then any edge $(u,v)$, such that both $u, v \in X$, is satisfied with probability at least $1/t^2$. By the {\em regularity} of the Label Cover instance, there are $\Omega(\alpha)$ fraction of edges in $X$ and hence the labeling satisfies at least $\Omega(\alpha/t^2)$ fraction of the total edges in $G$ in expectation. Setting $\alpha \approx L^{-\varepsilon/3}$ gives a contradiction. Thus, if we set $L$ such that $3^L \gg |V|$, there is an inverse logarithmic dependency between the size of the hypergraph ($\approx 3^L$) and the lower bound on the fractional size of the independent set $\Omega(1/L^{\varepsilon/3})$ which proves Theorem~\ref{theorem:three_four}.

The reduction and the analysis of $6$-uniform hypergraph construction (Theorem~\ref{theorem:two_six}) is exactly the same as before. Instead of $2$-wise $t$-agreeing family of $[3]^L$, we use $3$-wise $t$-agreeing family of $\{0,1\}^L$ (Theorem~\ref{thm:FT_two}). 

The reduction and analysis in the proofs of Theorems ~\ref{theorem:three_three} and ~\ref{theorem:two_four} are similar to each other where the starting point is the layered Label Cover instance.

\subsection{Comparison with previous works}
\label{section:comparison}

Previous reductions ~\cite{Khot02_33, Khot02_q4, GHHSV17, KS14} showing independent set guarantee require {\em smoothness} property of the Label Cover instance (See Section~\ref{section:lc} for the definition of Label Cover). The smoothness property roughly says that for any {\em small} list of labels to $u\in U$, these labels are projected to different labels, with high probability, if we choose a random constraint attached to $u$. The reason they need this property is rather technical and there is no intuitive reason for it. Saket's work~\cite{S14} uses following structural property of the Label Cover instance: There is a universal constant $c_0>0$, such that for any $u\in U$ and $S\subseteq [L]$
$$\Pr_{v\sim N(u)}[|\phi_{u\rightarrow v}(S)| < |S|^{c_0}] \leq |S|^{-c_0}.$$ 
 Therefore, all the previous proofs of Theorem~\ref{theorem:three_four} and ~\ref{theorem:three_three}  exploited the special structure of the Label Cover instance constructed from the PCP Theorem~\cite{AroraS1998, FGLSS, AroraLMSS1998} and the parallel repetition theorem of Raz~\cite{Raz98}.

In our proofs of Theorem~\ref{theorem:three_four}, and ~\ref{theorem:two_six}, we do not need any structure on the projection constraints of the Label Cover. In fact, for Theorem~\ref{theorem:three_four} and ~\ref{theorem:two_six}, we can start with a gap instance of $2$-CSP over alphabet $[L]$ and {\em arbitrary} constraints with completeness $1$ and soundness $\frac{1}{L^{\varepsilon}}$ for some $\varepsilon>0$.\footnote{to get the quasi \np-hardness, we would still need the reduction from $3$-SAT of size $n$ to the gap instance of $2$-CSP to run in time $n^{(\log L)^C}$ for some constant $C\geq 0$.} In the current exposition though, we start with the usual gap Label Cover instance to keep things simple. For our proofs of Theorem~\ref{theorem:three_three} and Theorem~\ref{theorem:two_four} though, we need the smoothness property of the layered  Label Cover instance.

Previous works including ~\cite{DRS02} and ~\cite{B18} which use agreement based decoding (like ours), only show hardness for approximating chromatic number. In fact, their hypergraphs always contain an independent set of size around half in both the completeness and soundness cases. Our reductions and analyses are very similar to ~\cite{DRS02}  and ~\cite{B18}, but we get independent set guarantee by using theorems on the extremal $t$-agreeing families.

We also like to point out that although our proofs are modular and require weaker conditions on the Label Cover instance (in two cases), the earlier reductions which use Fourier analysis often prove a stronger statement. More specifically, in the soundness case, we show that the maximum sized independent set is upper bounded by $\alpha n$, whereas the previous reductions usually prove a robust statement of the form -  every subset of vertices of size $\alpha n$ contains at least $f(\alpha)$ fraction of edges. We make no attempt to confirm the stronger soundness guarantees. Nonetheless, it would be interesting to know if this holds for our reduction.

\section{Organization}
We start with preliminaries first, where we define Label Cover and variants of it in Section~\ref{section:lc}.  We then state and prove results on the size of $k$-wise $t$-agreeing of $[q]^n$ in Section~\ref{sec:tagree}. In Section~\ref{section:log}, we prove a general Theorem ~\ref{thm:general 1}, where the starting point is the Label Cover instance. Theorem ~\ref{theorem:three_four} and ~\ref{theorem:two_six} follow as corollaries of Theorem ~\ref{thm:general 1}.

In Section~\ref{section:loglog}, we prove a general Theorem ~\ref{thm:general 2}, where the starting point is a multi-layered Label Cover instance. Theorem ~\ref{theorem:three_three} and ~\ref{theorem:two_four} follow as corollaries of Theorem ~\ref{thm:general 2}.

\section{Preliminaries}
We use $[q]$ to denote the set $\{1,2, \cdots, q\}$ and for a string $x\in [q]^n$, we use $x(i)$ to denote the element at its $i^{th}$ location.  We use $A \sqcup B$ to denote a disjoint union of two sets $A$ and $B$.

\subsection{Label Cover}
 \label{section:lc}
 A \emph{Label Cover instance} is given by a tuple $\mathcal{L} = (U,V,E,[L],[R],\Phi)$. The variables of the instance are $U \sqcup V$, with the variables in $U$ taking values in $[L]$ and the variables in $V$ taking values in $[R]$. We have a set of edges $E \subseteq U \times V$ and for every $(x,y) \in E$, there is a constraint $\phi_{x \rightarrow y} \in \Phi$. Moreover, these constraints are \emph{projection} constraints: this means (slightly abusing notation), that there is a map $\phi_{x \rightarrow y}:[L] \rightarrow [R]$ for every $(x,y) \in E$ such that, for every $i \in [L]$, $\phi_{x \rightarrow y}(i)$ is the unique assignment to $y$ that satisfies $\phi_{x \rightarrow y}$.
 
 We will use $\phi_{x \rightarrow y}$ to denote both the projection map as well as the constraint itself, when there is no ambiguity. For $A \subseteq U$ and $B \subseteq V$, we define
 \[
 \Phi(A,B) := \{\phi_{x \rightarrow y} \suchthat (x,y) \in (A \times B) \cap E\}.
 \]
 
 For an $x \in U$ and $A \subseteq V$, we define
 \[
 N_A(x) = \{y\suchthat (x,y) \in E~\text{and}~y \in A\}.
 \]
 In the case that $A = V$, we drop the subscript and denote $N(x) = N_V(x)$. These notations are symmetric and by a slight abuse of notation, will also be used when $x \in V$ and $A \subset U$. 
 
 We have the following theorem as a result of the PCP Theorem~\cite{AroraS1998, FGLSS, AroraLMSS1998} and the parallel repetition theorem of Raz~\cite{Raz98}.
 
\begin{theorem} 
\label{thm:labelcover}
There is a $C >0$ such that the following holds: For any parameter $b \in \N$, there exists a reduction from a $3$-SAT instance of $n$ variables to a Label Cover instance with at most $n^{O(b)}$ variables over an alphabet of size at most $2^{O(b)}$. The Label Cover instance has the following completeness and soundness conditions:
\begin{enumerate}
\item If the $3$-SAT instance is satisfiable, then there exists an assignment to the Label Cover instance that satisfies all the constraints.
\item If the 3-SAT instance is not satisfiable, then every assignment to the Label Cover instance satisfies at most a $2^{-Cb}$ fraction of the constraints.
\end{enumerate}
Moreover, the Label Cover instance produced is bi-regular and the reduction runs in time $n^{O(b)}$.
 \end{theorem}
 
 \subsubsection{Multi-layered Label Cover}
 
 We also have the multilayered Label Cover problem, which will be used as a starting point for proving Theorems~~\ref{theorem:three_three} and \ref{theorem:two_four}. An $\ell$-layered Label Cover instance is a tuple $\mathcal{L} = (\mathcal{U},\mathcal{R},E,\Phi)_{\ell}$. Here, $\mathcal{U} = \{U_i \suchthat i\in [\ell]\}$ where each $U_i$ is a set of variables. The variables in $U_i$ take values in $[R_i]$ and $\mathcal{R} = \{[R_i] \suchthat i \in [\ell]\}$ are the sets of labels. We also have a set of edges $E = \{E_{i,j} \subseteq U_i \times U_j\suchthat 1 \leq i < j \leq \ell\}$ and a set of constraints $\Phi$ where for $1 \leq i < j\leq \ell$, the constraints in $\Phi(U_i,U_j)$ project from the variables in $U_i$ to the variables in $U_j$.

 In using multi-layered Label Covers, we will require a bit more of structure on them.
 \begin{definition}[Weakly Dense]
  A multi-layered Label Cover instance $\mathcal{L} = (\mathcal{U},\mathcal{R},E,\Phi)_{\ell}$ is \emph{weakly dense} if for any $1 \leq m \leq \ell$,  any sequence of distinct integers $1 \leq i_1< \cdots \cdots i_m \leq \ell$, and any sequence of sets $S_j \subseteq U_{i_j}$ such that $|S_j| \geq \frac{2}{m}|U_{i_j}|$ for all $j \in [m]$, there are two sets $S_k$ and $S_{k'}$ such that $\Phi(S_k,S_{k'}) \geq \frac{1}{m^2}|\Phi(U_{i_k},U_{i_{k'}})|$.
 \end{definition}

 \begin{definition}[Smoothness]
  A multi-layered Label Cover instance $\mathcal{L} = (\mathcal{U},\mathcal{R},E,\Phi)_{\ell}$ is $T$-\emph{smooth} if for every $1 \leq i < j \leq \ell$, $x \in U_i$ and $S \subseteq [R_i]$, it holds that
  \[
  \Pr_{y\in N_{U_j}(x)}[|\phi_{x\rightarrow y}(S)| < |S|] \leq \frac{|S|^2\ell}{T}.
  \]
 \end{definition}

 We have the following theorem which gives us hardness for the smooth weakly dense multilayered Label Cover problem.
 
\begin{theorem}\cite{DRS02, DGKR05, Khot02_33}
\label{thm:multilayer}
For any parameters $T\gg \ell\geq 2,~r \in \N$, there exists a reduction from $3$-SAT instances of size $n$ to $\ell$-layered \emph{weakly dense} $T$-\emph{smooth} Label Cover instances with layers $U_1,\ldots, U_{\ell}$, and $n^{O((T+\ell) r)}$ variables over a range of size $2^{O(\ell r)}$. The Label Cover instance has the following completeness and soundness conditions:

\begin{enumerate}
\item If the $3$-SAT instance is satisfiable, then there exists an assignment to the Label Cover instance that satisfies all the constraints.
\item If the 3-SAT instance is not satisfiable, then for every $1 \leq i < j \leq \ell$, every assignment to the Label Cover instance satisfies at most $2^{-\Omega(r)}$ fraction of the constraints between layers $U_i$ and $U_j$.
\end{enumerate}
Moreover, the reduction runs in $n^{O((T+\ell) r)}$ time. 
\end{theorem}
 
 \subsection{Bounds on $t$-agreeing Families}
 \label{sec:tagree}

 For an alphabet $\Sigma$ and strings $a_1, \ldots, a_k \in \Sigma^n$ we define the \emph{agreement} $\agreement(a_1, \ldots, a_k)$ to be the set of coordinates where all $k$ strings agree, i.e.,
 \[
 \agreement(a_1, \ldots, a_k) = \{i \in [n] \suchthat a_1(i) = a_2(i) = \ldots = a_k(i)\}.
 \]

If $|\agreement(a_1, \ldots, a_k)| \geq t$, we say that $a_1, \ldots, a_k$ are \emph{$t$-agreeing}. We say that a family $\mathcal{F} \subseteq \Sigma^{n}$ is \emph{$k$-wise $t$-agreeing} if all subsets of $k$ strings $a_1, \ldots, a_k \in \mathcal{F}$ are $t$-agreeing.  For the special case $k=2$ we drop the ``$k$-wise'' and simply call the family $t$-agreeing for brevity.

We have the following bound on the maximal $t$-agreeing subfamily of $[3]^n$:
 
 \begin{theorem}[\cite{FT99, AK98}]
 \label{thm:FT}
 Let $n$ and $t$ be integers such that $n \geq 3t - 1$. Then for every $t$-agreeing family $\mathcal{F} \subset [3]^n$ it holds that
 \[
 |\mathcal{F}| \leq 3^{n - 3t +1} \sum_{i = 0}^{t-1}\binom{3t-1}{i}2^i.
 \]
 \end{theorem}
 
 It is easy to check that such a family $\mathcal{F}$ as described above has size at most $3^{n - t/10}$ for large enough $t$.

For the proofs of Theorems~\ref{theorem:two_six} and ~\ref{theorem:two_four}, we need a similar theorem, but for a family of subsets of $\{0,1\}^n$. Note that here, there are $t$-agreeing families of $\{0,1\}^n$ of size at least  $2^{n-1}(1 -o(1))$ for $t = o(\sqrt{n})$ (by taking all strings of Hamming weight $\le (n-t)/2$).  However, for the maximum size of a $3$-wise $t$-agreeing family we have a similar upper bound as Theorem~\ref{thm:FT}.

 \begin{theorem}
 \label{thm:FT_two}
 Let $n$ and $t$ be integers such that $n \geq t$.  Then for every $3$-wise $t$-agreeing $\mathcal{F} \subset \{0,1\}^n$ it holds that
 \[
 \frac{|\mathcal{F}|}{2^n} \leq \left(\frac{\sqrt{5}-1}{2}\right)^t.
 \]
 \end{theorem}

We say that a family $\mathcal{F} \subset \{0,1\}^{n}$ is {\em $3$-wise $t$-intersecting} if for every $a,b ,c \in \mathcal{F}$ it holds that $|\{ i\in [n] \mid a(i) = b(i) = c(i) = 1\}| \geq t$. 
 \begin{theorem}[\cite{F76}]
 \label{thm:Fint}
 Let $n$ and $t$ be integers such that $n \geq t$. Then for every $3$-wise $t$-intersecting $\mathcal{F} \subset \{0,1\}^n$ it holds that
 \[
 \frac{|\mathcal{F}|}{2^n} \leq \left(\frac{\sqrt{5}-1}{2}\right)^t.
 \]
 \end{theorem}

Theorem~\ref{thm:Fint} essentially follows\footnote{Frankl proves this for $t = 3$ but the proof can easily be seen to work by a minor modification of Proposition $3$ in~\cite{F76}.} from a beautiful proof of Frankl~\cite{F76} (see also~\cite{F18} where this statement is made explicit), on the size of $3$-wise $t$-{\em intersecting} families.  The proof of Theorem~\ref{thm:FT_two} now follows using a standard {\em shifting} argument.
 
 \begin{proof}[Proof of Theorem~\ref{thm:FT_two}]
 Let $\mathcal{F}\subseteq \{0,1\}^n$ be any $3$-wise $t$-agreeing family. Iteratively for every $i=1, 2, \cdots, n$ we do the following shifting of $\mathcal{F}_0 \defeq \mathcal{F}$ to get families $\mathcal{F}_1, \mathcal{F}_2, \cdots, \mathcal{F}_n$. If $x\in \mathcal{F}_{i-1}$ is such that $x(i) = 0$ and $x \oplus e_i \notin \mathcal{F}_{i-1}$\footnote{the operation $x\oplus e_i$ flips the $i^{th}$ bit of $x$.}, then in $\mathcal{F}_i$ we replace $x$ with $x\oplus e_i$ (and otherwise we keep $x$ in $\mathcal{F}_i$). 
 
 We have following two invariants after every iteration: (1) the size of the family remains unchanged. (2) the modified family is still a $3$-wise $t$-agreeing family. (1) is obvious. To see that (2) holds, we assume for contradiction that $\mathcal{F}_i$ is no longer a $3$-wise $t$-agreeing family, whereas $\mathcal{F}_{i-1}$ is a $3$-wise $t$-agreeing family. This means that there exists $a', b', c'\in \mathcal{F}_i$ such that $\agreement(a',b', c') < t$. Since in the $i^{th}$ iteration, only the $i^{th}$ location gets affected, it must be the case that $a'(i), b'(i), c'(i)$ are not all the same. Without loss of generality, assume $a'(i) = 0$ and $a'\in \mathcal{F}_{i-1}$. This means that $a' \oplus e_i \in \mathcal{F}_{i-1}$. Up to symmetry between $b'$ and $c'$ we have two cases to handle:
 
\begin{description}
\item[Case 1:] $b'\in \mathcal{F}_{i-1}$ and $c \in \mathcal{F}_{i-1}$ such that $c(i)=0$ and $c'= c\oplus e_i$.\\
 In this case, $\agreement(a',b', c') = \agreement(a'\oplus e_i,b', c) <t$, regardless of $b'(i)$, contradicting the fact that $\mathcal{F}_{i-1}$ is a $3$-wise $t$-agreeing family, as $\{a'\oplus e_i,b', c\}\subseteq \mathcal{F}_{i-1}$
 
\item[Case 2:] $b, c\in \mathcal{F}_{i-1}$ such that $b(i)= c(i) =0$ and $b' = b\oplus e_i$ and $c' = c\oplus e_i$.\\In this case, $\agreement(a'\oplus e_i,b, c) = \agreement(a',b', c') <t$. contradicting the fact that $\mathcal{F}_{i-1}$ is a $3$-wise $t$-agreeing family, as $\{a'\oplus e_i,b, c\}\subseteq \mathcal{F}_{i-1}$.
\end{description}

We keep reiterating the above $n$-step process until $\mathcal{F}_0 = \mathcal{F}_n$ (by letting $\mathcal{F}_0\leftarrow \mathcal{F}_n$ at the start). The process must halt, since at each $n$-step iteration if  $\mathcal{F}_0 \neq \mathcal{F}_n$, then we have increased the total hamming weights of strings in $\mathcal{F}_0$ by at least $1$.  If at the end of the $n$-step process, we have $\mathcal{F}_0 = \mathcal{F}_n$, then this condition means that $\mathcal{F}_n$ is a {\em monotone}  $3$-wise $t$-agreeing family. \footnote{In fact, after just the first $n$-step iteration, the family becomes monotone.}

We now claim that $\mathcal{F}_n$ is in fact a $3$-wise $t$-{\em intersecting} family. Suppose not, this means that there are $a,b,c \in \mathcal{F}_n$ such that $S_1 = \{ i\in [n] \mid a(i) = b(i) = c(i) = 1\}$ and $|S_1|<t$. 
Let $S_0 = \{ i\in [n] \mid a(i) = b(i) = c(i) = 0\}$.
Now, the monotonicity of $\mathcal{F}_n$ gives that there exists $a'\in \mathcal{F}_n$ such that $a'(i) = a(i)$ for all $i\in [n]\setminus S_0$ and $a'(i) = 1$ for all $i\in S_0$. We can conclude that $\agreement(a',b, c) = |S_1| <t $, contradicting the fact that $\mathcal{F}_n$ is a $3$-wise $t$-agreeing family.

Using Theorem~\ref{thm:Fint}, we then have
\[
 \frac{|\mathcal{F}_n|}{2^n} \leq \left(\frac{\sqrt{5}-1}{2}\right)^t,
 \]
and since $|\mathcal{F}| = |\mathcal{F}_n|$, the theorem follows.
\end{proof}
 
 \subsection{Other combinatorial lemmas}
 
 We will use another simple combinatorial lemma:
  
 \begin{lemma}
 \label{lem:star}
 Let $\mathcal{F}$ be a family of multisets of subsets of $[n]$ of size at most $t$. Suppose for every $S_1,\ldots,S_d \in \mathcal{F}$, there are distinct $i,j \in [t]$ such that $S_i \cap S_j \neq \emptyset$, then there exists an $i \in [n]$ such that 
 \[
 |\{S \in \mathcal{F} \suchthat i \in S\}| \geq \frac{1}{t(d-1)}|\mathcal{F}|.
 \]
 \end{lemma}
 \begin{proof}
   Let $S_1, \ldots, S_{r} \in \mathcal{F}$ be a maximal pairwise disjoint subfamily of $\mathcal{F}$, and note that $r \le d-1$.  Consider $S \defeq \bigsqcup_{i \in [r]} S_i$, and for $s \in S$ let $\mathcal{F}_s \defeq \{T \in \mathcal{F} \suchthat s \in T\}$. Because $S_1,\ldots,S_{r}$ form a maximal pairwise disjoint family, we have that $\mathcal{F} = \cup_{s \in S}\mathcal{F}_s$. On the other hand $|S| = tr \leq t(d-1)$. Therefore, there is some $s \in S$ such that $|\mathcal{F}_s| \geq \frac{1}{t(d-1)}|\mathcal{F}|$.
 \end{proof}
 
 In particular, setting $d = 2$ in the above lemma, we have the following simple corollary:
 
 \begin{corollary}
 \label{lem:specialstar}
Let $\mathcal{F}$ be a family of multisets of subsets of $[n]$ of size at most $t$ that is also intersecting, then there exists an $i \in [n]$ such that 
 \[
 |\{S \in \mathcal{F} \suchthat i \in S\}| \geq \frac{1}{t}|\mathcal{F}|.
 \]
 \end{corollary}

 \section{Poly-logarithmic hardness with large uniformity}
\label{section:log}
 \begin{theorem}
   \label{thm:general 1}
   Let $k$, $q$, and $c$ be such that for some $\delta = 1/\poly m$
   and all sufficiently large $m$, no family of $[q]^m$ of size at
   least $\delta \cdot q^m$ is $k$-wise $\delta^{-c}$-agreeing.
   Then $\hc_n(2k, q, \poly \log n)$ is quasi \nph.
 \end{theorem}

 As immediate corollaries, it follows by plugging in
 Theorem~\ref{thm:FT} that $\hc(4, 3, \poly \log n)$ is quasi \nph
 (Theorem~\ref{theorem:three_four}), and plugging in Theorem~\ref{thm:FT_two}
 that $\hc(6, 2, \poly \log n)$ is quasi \nph (Theorem~\ref{theorem:two_six}).
 We note that those Theorems on $t$-agreeing families in fact give
 bounds of the form $O(\log 1/\delta)$ on the amount of agreement,
 much better than the $\poly 1/\delta$ needed by
 Theorem~\ref{thm:general 1}.

 We now proceed to prove the theorem and begin by describing the reduction.
 Consider a Label Cover instance $\mathcal{L} = (U, V, E, [R], [L],\Phi)$. We reduce it to a hypergraph $\calH = (\HV,\HE)$ whose vertices and edges are as follows:

 \begin{description}
 \item[Vertices $\HV$:] The vertex set is obtained by replacing each variable $x \in U$ by a cloud $[q]^L$ of vertices: for a variable $x \in U$, denote
   \[
   \HV_x := \{(x,a) \suchthat a \in [q]^{L}\}.
   \]
   The vertex set of $\mathcal{H}$ is given by
   \[
   \HV= \bigcup_{x \in U} \HV_x.
   \]
 \item[Edges $\HE$:] For every $y \in V$ and $x_1,x_2 \in N(y)$, there is a hyperedge on a set of $2k$ vertices $\{(x_1,a_1), \ldots, (x_1, a_k), (x_2,b_1), \ldots, (x_2, b_k)\}$ if they have the property that for every $i_1,i_2 \in [L]$ such that  $\phi_{x_1 \rightarrow y}(i_1) = \phi_{x_2 \rightarrow y}(i_2)$, it holds that 
   \[
   \left| \bigcup_{j=1}^k \left\{ \, a_j(i_1), \, b_j(i_2) \, \right\} \right| \geq 2.
   \]
 \end{description}

\subsection{Completeness}

\begin{lemma}
\label{lem:2k-unif-completeness}
If there is an assignment to the variables of $U \sqcup V$ that satisfies all the constraints in $\Phi$, then $\chi(\mathcal{H}) \leq q$.
\end{lemma}

\begin{proof}
Let $A:U \sqcup V \rightarrow [L] \cup [R]$ be the assignment that satisfies all the constraints of $\mathcal{L}$. Consider the coloring that colors vertex $(x,a) \in U \times[q]^{L}$ with the color $a(A(x))$. Suppose for contradiction that this yields a monochromatic hyperedge $\{(x_1,a_1), \ldots, (x_1, a_k), (x_2,b_1), \ldots, (x_2, b_k)\}$ for some $x_1,x_2 \in N(y)$ and some $y \in V$. 

Monochromaticity implies that
\[
\left| \cup_{j=1}^k \{ a_j(A(x_1)),b_j(A(x_2)) \} \right| = 1.
\]
However, since $A$ satisfies both $\phi_{x_1 \rightarrow y}$ and $\phi_{x_2 \rightarrow y}$, we also have that
\[
\phi_{x_1 \rightarrow y}(A(x_1)) = \phi_{x_2 \rightarrow y}(A(x_2)) = A(y).
\] 
Taken together, these contradict the condition for being an edge. 
\end{proof}

\subsection{Soundness}

 \begin{lemma}
   \label{lem:2k-unif-soundness}
   If $\alpha(\mathcal{H}) > 2 \delta$ then there exists an assignment to
   $U \sqcup V$ that satisfies at least a $\delta^{1+2c}$ fraction of the constraints $\Phi$
   (where $\delta$ and $c$ are as in Theorem~\ref{thm:general 1}).
 \end{lemma}

 \begin{proof}
 Let $\mathcal{I}$ be an independent set in $\mathcal{H}$ of size $|\mathcal{I}| \geq 2\delta |\HV|$. For every variable $x \in U$, let $\mathcal{I}_x \defeq \mathcal{I} \cap \HV_x$.  By an averaging argument there exists an $X \subseteq U$ such that:
 \begin{enumerate}
 \item $|X| > \delta |U|$.
 \item $|\mathcal{I}_x| >  \delta |\HV_x|$ for every $x \in X$.
 \end{enumerate}
 We will henceforth restrict our attention to variables in $X$.

 By the Theorem assumption that all $k$-wise $\delta^{-c}$-agreeing families of
 $[q]^n$ have size at most $\delta q^n$, it follows that for every $x \in X$, there are $k$ vertices $(x, a_{x,1}), \ldots, (x,a_{x,k})$ in $\mathcal{I}_x$ such that $|\agreement(a_{x,1}, \ldots, a_{x,k})| < \delta^{-c} \defeq t$.  Define $L_x = \agreement(a_{x,1}, \ldots, a_{x,k}) \subseteq [L]$.
     
  Another observation is that for $y \in V$,
  \begin{equation}
  \label{eqn:2k intersect}
  \phi_{x_1\rightarrow y}(L_{x_1}) \cap \phi_{x_2 \rightarrow y}(L_{x_2}) \neq \emptyset~\text{for every }x_1,x_2\in N_X(y).
  \end{equation}
  Indeed, if this were not the case, one can check that 
  \[\{(x_1,a_{x_1,1}), \ldots, (x_1, a_{x_1,k}), (x_2, a_{x_2,1}), \ldots, (x_2, a_{x_2,k}) \}\] 
  is a hyperedge, contradicting our assumption that these vertices are from an independent set.  

  We now define a (randomized) labelling $A: U \sqcup V \rightarrow
  [L] \cup [R]$ of $\mathcal{L}$.  For $x \in X$, pick a label $A(x)$
  from $L_x$ at random.
  
  For $y \in N(X)$, define the (multi-)family $\mathcal{F}(y) \defeq \{\phi_{x \rightarrow y}(L_x) \suchthat x \in N_X(y)\}$ of subsets of $[R]$. By~(\ref{eqn:2k intersect}), $\mathcal{F}(y)$ is an intersecting family where every set has size at most $t$. Therefore, Lemma~\ref{lem:specialstar} implies that
  there is a label $i \in [R]$ that is present in at least $\frac{1}{t} \cdot |\mathcal{F}(y)|$ sets in $\mathcal{F}(y)$.
    Define $A(y)$ to be that label.  For the remaining variables $x \in
  U \setminus X$ and $y \in V \setminus N(X)$, assign a label
  arbitrarily.

  By the choice of $A(y)$, for every $y \in N(X)$, a $\frac{1}{t}$
  fraction of all $x \in N_X(y)$ have a label $i \in L_x$ such that
  $\phi_{x \rightarrow y}(i) = A(y)$, and each such constraint is satisfied by $A$ with
  probability at least $1 / |L_x| \ge 1/t$.  It follows that the expected number of constraints satisfied by the labelling $A$ is at least $\frac{1}{t^2} |\Phi(X, V)|$.

  By the
  regularity of the Label Cover instance, $|\Phi(X, V)| =
  \frac{|X|}{|U|} |\Phi| = \delta |\Phi|$.
  Thus $A$ satisfies (in expectation) at least a $\delta / t^2 = \delta^{1+2c}$ fraction of all constraints.
  This proves the existence of an assignment that achieves the above guarantee.
 \end{proof}
 
 \begin{proof}[Proof of Theorem~\ref{theorem:three_four}]
   Start with a $3$-SAT instance $\Pi$ on $n$ variables. Setting $b =
   \log \log n$,
   Theorem~\ref{thm:labelcover} gives us a Label Cover instance
   $\mathcal{L} = (U,V,E,[L],[R],\Phi)$ with $n^{O(\log \log n)}$
   variables taking values over an alphabet of size at
   most $2^{O(b)} = \poly \log n$. Applying the reduction above gives us a
   hypergraph $\calH$ on $|\HV| = q^{\poly \log n} \cdot n^{O(b)} =
   2^{\poly \log n}$ vertices. In the completeness case,
   Lemma~\ref{lem:2k-unif-completeness}, gives us $\chi(\calH) \leq
   q$. On the other hand, in the soundness case, then no assignment to
   $\mathcal{L}$ will satisfy more than $2^{-Cb} = (\log n)^{-C}$ fraction of
   the constraints, where $C$ is the constant from
   Theorem~\ref{thm:labelcover}, and so by
   Lemma~\ref{lem:2k-unif-soundness}, $\alpha(\calH) \leq (\log
   n)^{-C/(1+2c)} = 1 / \poly \log |\HV|$.
 \end{proof}
 
 \begin{remark}
 In the analysis, we never use the fact that the constraints are projection constraints. Thus, one can start with any gap $2$-CSP over $[L]$ and carry out the reduction.
 \end{remark}

 \section{$\poly \log \log$-hardness with smaller uniformity}
 \label{section:loglog}
 
 In this section, we prove a following general theorem.
 \begin{theorem}
   \label{thm:general 2}
   Let $k$, $q$, and $c$ be such that for some $\delta = 1/\poly m$
   and all sufficiently large $m$, no family of $[q]^m$ of size at
   least $\delta \cdot q^m$ is $k$-wise $(c \log(1/\delta))$-agreeing.
   Then $\hc_n(k + 1, q, \tilde{\Omega}(\sqrt{\log \log n}))$ is quasi \nph.
 \end{theorem}

 As immediate corollaries, it follows using
 Theorem~\ref{thm:FT} that $\hc(3, 3, \tilde{\Omega}(\sqrt{\log \log n}))$ is quasi \nph
 (Theorem~\ref{theorem:three_three}), and using Theorem~\ref{thm:FT_two}
 that $\hc(4, 2, \tilde{\Omega}(\sqrt{\log \log n}))$ is quasi \nph (Theorem~\ref{theorem:two_four}).

 For this reduction, we start with an $\ell$-layered Label Cover $\mathcal{L} = (\mathcal{U},\mathcal{R},E,\Phi)_{\ell}$, where $\mathcal{U} = \{U_i \suchthat i\in [\ell]\}$ is the set of layers, $\mathcal{R} = \{[R_i] \suchthat i \in [\ell]\}$ are the sets of labels, $\Phi = \{\Phi_{i \rightarrow j} \suchthat i,j \in [\ell]\}$ are the sets of constraints. We reduce it to a hypergraph $\calH = (\HV,\HE)$ whose vertices and edges as follows:

 \begin{description}
   \item[Vertices $\HV$:] The vertex set is obtained by replacing each variable $x$ in each layer $U_i$ by a cloud $[q]^{R_i}$ of vertices: for a variable $x \in U_i$, denote
     \[
     \HV_x := \{(x,a) \suchthat a \in [q]^{R_i}\}.
     \]
     The vertex set is 
     \[
     \HV = \bigcup_{U_i \in \mathcal{U}} \bigcup_{x \in U_i} \HV_x.
     \]
          
   \item[Edges $\HE$:] For every $1 \le i < j \le \ell$, $y \in U_j$ and $x \in N_{U_i}(y)$, $k+1$ vertices 
     \[\{(x,a_1),(x,a_2),\ldots, (x, a_k), (y,b)\}
     \]
     forms a hyperedge if for every $r \in [R_i]$, it holds that 
     \begin{equation}
       \label{eqn:k+1 unifedge}
       |\{a_1(r),a_2(r), \ldots, a_k(r), b({\phi_{x \rightarrow y}(r)})\}| \ge 2.
     \end{equation}
 \end{description}

 In what follows we write $U = \bigsqcup_{i=1}^{\ell} U_i$ to denote the set of all variables of $\mathcal{L}$.
 
\subsection{Completeness}

\begin{lemma}
\label{lem:k+1 unifcompleteness}
If the there is an assignment to the variables of $U$ that satisfy all the constraints of $\Phi$, then the $\chi(\mathcal{H}) \leq q$.
\end{lemma}

\begin{proof}
Let $A: U \rightarrow \bigcup_{i \in [\ell]}[R_i]$ be an assignment that satisfies all the constraints in $\Phi$. A proper $q$-coloring is given by coloring a vertex $(x,a) \in U_i \times [q]^{R_i}$ with the color $a(A(x))$.  Suppose for contradiction that this is not a proper coloring and that there is a monochromatic edge $\{(x,a_1),\ldots,(x,a_k),(y,b)\}$ for some $y \in U_j$ and $x \in N_{U_i}(y)$.

Since $A$ satisfies $\phi_{x \rightarrow y}$, we have $\phi_{x\rightarrow y}(A(x)) = A(y)$. However, monochromaticity implies that 
\[
|\{a_1(A(x)), \ldots, a_2(A(x)), b(A(y))\}| = 1
\]
This contradicts the condition for $\{(x,a_1), \ldots, (x, a_k),(y,b)\}$ being a hyperegde.
\end{proof}

\subsection{Soundness}

 \begin{lemma}
   \label{lem:k+1 unifsoundness}
   Let $k$, $q$, $c$ and $\delta$ be as in Theorem~\ref{thm:general 2}, $t = c \log(1/\delta)$, $\ell = 2/\delta^2$, and suppose that $\mathcal{L}$ is $T$-smooth for some $T \ge \frac{16(c \log 1/\delta)^2\ell}{\delta^2}$.  Then if $\alpha(\mathcal{H}) \ge 4\delta$ then there exists $1 \le i < j \le \ell$ and an assignment to $\mathcal{L}$ which satisfies at least a $\frac{\delta^{4+2qc}}{\poly \log \delta}$ fraction of $\Phi_{i \rightarrow j}$.
 \end{lemma}
 
 \begin{proof}
   
   Let $\mathcal{I}$ be an independent set of size $|\mathcal{I}| \geq 4\delta|\HV|$. For $x \in U$, let us denote $\mathcal{I}_x := \mathcal{I} \cap \HV_x$. An averaging argument gives us that there is a $U' \subseteq U$ such that:
 
 \begin{enumerate}
 \item $|U'| \geq 2 \delta |U|$.
 \item $|\mathcal{I}_x| \geq 2 \delta |\HV_x|$ for $x \in U'$.
 \end{enumerate}
 A similar averaging argument gives us that there is a set $\mathcal{W} = \left\{U_{i_1}, \ldots, U_{i_{|\mathcal{W}|}}\right\}$ of layers with the following properties:
 
 \begin{enumerate}
 \item $|\mathcal{W}| =  \delta \ell$.
 \item For every layer $W \in \mathcal{W}$, $|W \cap U'| \geq \delta |W|$.
 \end{enumerate}
 For $j \in [|\mathcal{W}|]$, let us denote $Z_j \defeq U_{i_j} \cap
 U'$.  From the two properties of $\mathcal{W}$ we have $|Z_j| \ge
 \delta |U_{i_j}| = \frac{\delta^2 \ell}{|\mathcal{W}|} |U_{i_j}|$.
 By assumption $\ell = 2/\delta^2$ and thus $|Z_j| \ge
 \frac{2}{|\mathcal{W}|} |U_{i_j}|$, so from the Weakly Dense property
 of $\mathcal{L}$ it follows that there are $Z_{k} \subseteq U_{i_k}$ and $Z_{k'} \subseteq
 U_{i_{k'}}$ such that
 
 \begin{equation}
 \label{eqn:k+1 WeaklyDense}
 |\Phi(Z_k,Z_{k'})| \geq \frac{1}{|\mathcal{W}|^2}\cdot |\Phi(U_{i_{k}},U_{i_{k'}})| = \frac{\delta^2}{4}  |\Phi(U_{i_{k}},U_{i_{k'}})|.
 \end{equation}
 By yet another averaging argument, at least a $\frac{\delta^2}{8}$ fraction of $x\in Z_k$ has at least a $\frac{\delta^2}{8}$ fraction of their constraints in $Z_{k'}$. Let us denote those $x$ by $X\subseteq Z_{k}$ and let $Y \defeq Z_{k'}$. We will henceforth restrict our attention to the constraints between $X$ and $Y$. We have that for every $x \in X$, $\mathcal{I}_x$ there are $k$ vertices, $(x,a_{x,1}), \ldots, (x, a_{x,k})$ such that $|\agreement(a_{x_1}, \ldots, a_{x, k})| \le c \log 1/\delta \defeq t  $. Let us denote $L_x = \agreement(a_{x,1}, \ldots, a_{x,k})$. 
 
 We now trim bad constraints from $\Phi(X, Y)$ as follows. For every $x\in X$, we have from the smoothness property that at most $\frac{t^2 \ell}{T} \le \frac{\delta^2}{16}$ fraction of constraints $\phi_{x\rightarrow y}$ where $y\in  U_{i_{k'}}$ are such that  $|\phi_{x\rightarrow y}(L_x)| < |L_x|$; Call such constraints {\em bad}. We remove all such bad constraints from $\Phi(X, Y)$. This gives that at least $\frac{\delta^2}{16}$ fraction of the $y\in N_{Y}(x)$ are such that $|\phi_{x\rightarrow y}(L_x)| = |L_x|$, i.e., no two labels from $L_x$ project on the same label and we retain all such constraints.  We use $N_X'(y)$ to denote only the neighbors with no bad constraints. Thus we are left with $\frac{\delta^2}{4}\cdot \frac{\delta^2}{16}\cdot |\Phi(U_{i_{k}},U_{i_{k'}})|$ constraints in $\Phi(X, Y)$. Let us call these $\Phi_{\text{good}}$.

 \begin{claim}
   \label{claim:soundness intersecting family}
   For any $y \in Y$, consider any set $\{x_1,\ldots, x_{s}\} \subseteq N_X'(y)$ such that for every $i, j \in [s]$, we have $\phi_{x_i \rightarrow y}(L_{x_i}) \cap \phi_{x_j \rightarrow y}(L_{x_j}) = \emptyset$. Then
   \[
   s \le \frac{\log(1/\delta)}{\delta^{2qc}}
   \]
 \end{claim}
  
 \begin{proof}
    For every $i \in [s]$ and $(y,b) \in \mathcal{I}_y$, condition~(\ref{eqn:k+1 unifedge}) implies that $b(\cdot)$ at locations $\cup_{r \in L_{x_i}}\phi_{x_i \rightarrow y}(r)$  can take at most $q^{|L_{x_i}|} - (q-1)^{|L_{x_i}|}$ different values. This is because if all values at $b(\phi_{x_i \rightarrow y}(r))$ are from $[q]\setminus a_{x_i,1}(r)$ (which by definition equals $[q]\setminus a_{x_i,j}(r)$ for all $j \in [k]$), then $(x_i, a_{x_i,1}), \ldots,$ $ (x_i, a_{x_i,k}), (y, b)$ forms a hyperedge in $\mathcal{H}$, which is a contradiction as they all belong to $\mathcal{I}$. Since $\phi_{x_i \rightarrow y}(L_{x_i}) \cap \phi_{x_j \rightarrow y}(L_{x_j}) = \emptyset$, these restrictions for each $i\in [s]$ are disjoint from each other. Given this, we can upper bound the size of $\mathcal{I}_y$ by
 \[
 |\mathcal{I}_y| \leq \left(1 - \left(\frac{q-1}{q}\right)^t\right)^{s}|\HV_y| \le \exp(-s \exp(-2qt)) |\HV_y| = \exp(-s \delta^{2qc}) |\HV_y|.
 \] 
 
 On the other hand, we have that $|\mathcal{I}_y| \geq \delta |\HV_y|$. Combining these two facts gives us our desired bound.
 \end{proof}

 The rest of the proof proceeds exactly as the proof of Lemma~\ref{lem:2k-unif-soundness}.
 We now define a (randomized) labelling
 $A: X \sqcup Y \rightarrow [R_{i_k}] \cup [R_{i_{k'}}]$.  For $x \in X$, pick a
 label $A(x)$ from $L_x$ at random.

 For $y \in Y$, define the (multi-)family
 $\mathcal{F}(y) \defeq \{\phi_{x \rightarrow y} (L_x)\suchthat x \in N_X'(y)\}$ of subsets of $[R_{k'}]$.  By
 Claim~\ref{claim:soundness intersecting family} and Lemma~\ref{lem:star}, it follows that there is a label $\ell_y \in [R_{i_{k'}}]$ that is present in at least a $\frac{1}{t(s-1)} |\mathcal{F}(y)|$ sets in $\mathcal{F}(y)$.  Define $A(y)$ to be that label.

  By the choice of $A(y)$, for every $y \in N'(X)$, a $\frac{1}{t(s-1)}$
  fraction of all $x \in N_X(y)$ have a label $\ell_x \in L_x$ such that
  $\phi_{x \rightarrow y}(i) = A(y)$, and each such constraint is satisfied by $A$ with
  probability at least $1 / |L_x| \ge 1/t$.  It follows that the expected number of constraints satisfied by the labelling $A$
  is at least $\frac{1}{t^2(s-1)} |\Phi_{\text{good}}|$.

  As noted earlier,
  $|\Phi_{\text{good}}| \ge \frac{\delta^4}{64} |\Phi(U_{i_k}, U_{i_{k'}})|$
  and thus $A$ satisfies at least a
  $\frac{\delta^{4+2qc}}{\poly \log \delta}$ fraction of all
  constraints between layers $U_{i_k}$ and $U_{i_k'}$.
 \end{proof}
 
 \begin{proof}[Proof of Theorem~\ref{thm:general 2}]
   Start with a $3$-SAT instance $\Pi$ on $n$ variables. Set $\ell = c_1 \frac{\log\log n}{\log \log \log n}$, $T = (\log \log n)^3$ and $r = c_2 \log \log \log n$. Theorem~\ref{thm:multilayer} gives us an $\ell$-layered Label Cover instance $\mathcal{L} = (\mathcal{U}, \mathcal{R}, E, \Phi)$ on $n^{O((T+\ell)r)} = n^{\poly \log \log n}$ vertices over alphabets of size $2^{O(\ell r)} = 2^{O(c_1c_2 \log \log n)}$ with soundness $2^{-\Omega(r)}$, and we choose the constants $c_1$ and $c_2$ such that the alphabet size is $\log n$ and the soundness is a sufficiently large power of $\log \log n$ to apply Lemma~\ref{lem:k+1 unifsoundness} below.

The above reduction gives us a $(k+1)$-uniform hypergraph $\calH$ on $|\HV| = 3^{\log n}\cdot n^{\poly \log \log n}= n^{\poly \log \log n}$ vertices and in particular $\log \log |\HV| = O(\log \log n)$. In the completeness case, we have $\chi(\calH) \leq q$. Setting $\delta = \sqrt{2/\ell} = (\log \log n)^{-1/2 + o(1)}$ in the soundness case, we get that from  Lemma~\ref{lem:k+1 unifsoundness} that $\alpha(\calH) \leq (\log \log n)^{-1/2 + o(1)} = (\log \log |\HV|)^{-1/2 + o(1)}$.
 \end{proof}

\section{Conclusion}
\label{section:conclusion}

 We hope that although we do not improve upon previous results, our proof technique will be useful in improving hypergraph coloring hardness for lower uniformity. The reduction has an outer verifier (Label Cover instance) and an inner verifier (gadget) as two components. It might be the case that our inner verifier is stronger than the previous inner verifiers and hence we do not require extra structural property on the Label Cover instance (for Theorem~\ref{theorem:three_four}, and ~\ref{theorem:two_six}). It will be interesting to understand this trade-off.

 It is interesting that all of our results follow from a general framework and uses t-agreeing families. This can be thought as a unified proofs for results which otherwise had somewhat different proofs. We hope that this sheds light on the hardness of hypergraph coloring results regardless of its uniformity and the completeness guarantee.

 We now give a concrete open problem which will improve the proven hardness results in this paper. Consider the following property of a family $\mathcal{F} \subset [q]^n$, parameterized by $t$ and $T = T(t)$: For any subset $S \subseteq \mathcal{F}$ such that $|S| \geq \frac{1}{T}|\mathcal{F}|$ there exists $a,b \in S$ such that $|\agreement(a,b)| \leq t$. We say that such a family $\mathcal{F}$ has property $\mathfrak{P}(t,T)$. Theorem~\ref{thm:FT} shows that $[3]^n$ satisfies $\mathfrak{P}(t, 2^{O(t)})$ for any $t$. We pose the following problem, which, to the best of our knowledge has not been investigated yet:
 
 \paragraph{Open Problem:} For $q\geq 3$, is there a family $\mathcal{F} \subset [q]^n$ such that $|\mathcal{F}| \leq q^{\poly(\log n)}$ that has property $\mathfrak{P}(t, t^{ \omega(1)})$?

\bibliographystyle{alpha}
\bibliography{refs}

\end{document}

%% file: table.tex
\begin{figure}
\begin{tabular}{|c|c|l|}
\hline
Uniformity $k$ & Colors $c$ & Hardness $D$ \\
\hline
$3$ & $3$ & $\Omega((\log \log n)^{1/9})$ \cite{Khot02_33} \\
$3$ & $3$ & $2^{\frac{\log \log n}{\log \log \log n}}$ \cite{GHHSV17} \\
$3$ & $2$ & $\Omega((\log \log n)^{1/9})$ \cite{DRS02} *no independent set guarantee \\
$3$ & $3$ & $\Omega((\log \log n)^{1/2})$ [this paper] \\
\hline
$4$ & $2$ & $\Omega(\frac{\log\log n}{\log \log \log n})$ \cite{H02, GHS02} \\
$4$ & $q \ge 7$ & $(\log n)^{\Omega(q)}$ \cite{Khot02_q4} \\
$4$ & $2$ & $(\log n)^{\Omega(1)}$ \cite{S14} \\
$4$ & $4$ & $2^{2^{\sqrt{\log \log n}}}$ \cite{GHHSV17} \\
$4$ & $4$ & $2^{(\log n)^{\Omega(1)}}$ \cite{Varma15} \\
$4$ & $2$ & $(\log n)^{\Omega(1)}$ \cite{B18} *no independent set guarantee \\
$4$ & $2$ & $\Omega((\log n \log n)^{1/2})$ [this paper] \\
$4$ & $3$ & $(\log n)^{\Omega(1)}$ [this paper] \\
\hline
$6$ & $2$ & $(\log n)^{\Omega(1)}$ \cite{DG15} \\
$6$ & $2$ & $(\log n)^{\Omega(1)}$ [this paper] \\
\hline
$8$ & $2$ & $2^{2^{\sqrt{\log \log n}}}$ \cite{GHHSV17} \\
$8$ & $2$ & $2^{(\log n)^{\Omega(1)}}$ \cite{Varma15} \\
$8$ & $2$ & $2^{(\log n)^{1/10-o(1)}}$ \cite{Huang15} \\
\hline
$12$ & $2$ & $2^{(\log n)^{\Omega(1)}}$ \cite{KS14} \\
\hline
\end{tabular}
\caption{Comparison of the known results  and our results on the hardness of hypergraph coloring.}
\label{fig:comparison}
\end{figure}